\documentclass[12pt,letterpaper]{article}

\usepackage[utf8]{inputenc}
\usepackage[T1]{fontenc}
\usepackage{amsmath,amssymb,amsthm}
\usepackage{mathtools}
\usepackage{graphicx}
\usepackage{booktabs}
\usepackage{array}
\usepackage{multirow}
\usepackage{threeparttable}  % For table notes
\usepackage{caption}
\usepackage{subcaption}
\usepackage{float}
\usepackage{setspace}
\usepackage[margin=1in]{geometry}
\usepackage{natbib}
\usepackage{hyperref}
\usepackage{xcolor}
\usepackage{tikz}
\usetikzlibrary{arrows.meta,positioning,calc}

\newtheorem{theorem}{Theorem}
\newtheorem{lemma}[theorem]{Lemma}
\newtheorem{proposition}[theorem]{Proposition}

\newtheorem{definition}{Definition}
\newtheorem{assumption}{Assumption}

\newcommand{\E}{\mathbb{E}}
\newcommand{\R}{\mathbb{R}}

\newcommand{\Prob}{\mathbb{P}}
\DeclareMathOperator*{\argmax}{arg\,max}

\title{\Large \textbf{Ecosystem Competition and Cross-Market Subsidization: A Dynamic Theory of Platform Pricing}\thanks{I thank participants at various seminars for valuable comments.}}

\author{
Liang Chen
}

\date{January 2026}

\begin{document}

\maketitle

\begin{abstract}
\noindent Platform giants in China have operated with persistently compressed margins in highly concentrated markets for much of the past decade, despite market shares exceeding 60\% in core segments. Standard theory predicts otherwise: either the weaker firm exits, or survivors raise prices to monopoly levels. We argue the puzzle dissolves once firms are viewed as ecosystem optimizers rather than single-market profit maximizers. We develop a dynamic game in which a firm's willingness to subsidize depends on the spillover value its users generate in adjacent markets---what we call \textit{ecosystem complementarity}. When this complementarity is strong enough, perpetual below-cost pricing emerges as the unique stable equilibrium. The result is not predation in the classical sense; there is no recoupment phase. It is a permanent state of subsidized competition, rational for each firm individually but potentially inefficient in aggregate. We characterize the equilibrium, establish its dynamic stability, and show that welfare losses compound over time as capital flows into subsidy wars rather than innovation. The model's predictions are consistent with observed patterns in Chinese platform markets and suggest that effective antitrust intervention should target cross-market capital flows rather than prices.

\bigskip

\noindent \textbf{arXiv:} CS.IT; cs.GT; econ.TH

\medskip


\medskip

\end{abstract}

\newpage
\onehalfspacing

\section{Introduction}

Chinese platform firms have operated with persistently compressed margins in highly concentrated markets for the better part of a decade. This is puzzling. The textbook prediction is unambiguous: price wars end when either the weaker firm exits or the survivors coordinate on higher prices \citep{tirole1988theory}. Neither has happened here. The market has consolidated into tight duopoly, yet margins remain compressed and subsidies continue. Something in the standard framework is missing.

The Chinese food delivery market illustrates the pattern. Meituan and Alibaba's Ele.me together control roughly 95\% of the market, with Meituan holding a 67\% share.\footnote{Market share data from industry reports; see Section 4 for detailed sources.} By any textbook measure, this is a mature duopoly with substantial pricing power. Yet pure delivery margins remain modest---typically 10-15\% for market leaders---far below what standard models predict for dominant platforms with 60\%+ market share. When JD.com entered in early 2025, it triggered a subsidy war costing an estimated RMB 100 billion in six months. Daily orders reached 150 million at peak. None of the three firms has shown any intention of exiting or raising prices.

The standard toolkit from industrial organization struggles with this pattern. Predatory pricing, as \citet{milgrom1982predation} and others have analyzed it, presupposes a clear endgame: drive out the rival, then recoup losses through monopoly rents. But recoupment logic fails here. These firms are not trying to eliminate each other; they are locked in a permanent war of subsidization with no exit in sight. Something else must be going on.

This paper argues that the answer lies in what happens \textit{outside} the contested market. A Meituan user who orders lunch becomes a potential borrower on Meituan Finance. An Ele.me transaction generates data that feeds Alibaba's credit-scoring algorithm. The delivery business operates at thin margins, serving as a user-acquisition channel for adjacent services where returns are higher. The firm maximizes the value of its portfolio, not profit in any single market. We formalize this logic under the label \textit{ecosystem competition}.

The central object in our model is the \textit{ecosystem complementarity function} $\Psi(\cdot)$, which maps primary market share to spillover value in adjacent markets. When the marginal complementarity $\Psi'(m)$ exceeds a critical threshold, the calculus of competition changes in a discontinuous way. Each additional user generates more value through cross-market externalities than the subsidy costs to acquire them. Below-cost pricing becomes individually rational---not as a temporary investment, but as a permanent equilibrium strategy.

The analysis yields three sets of results. First, we characterize the Markov Perfect Equilibrium and show that perpetual subsidization emerges when ecosystem complementarity exceeds a threshold. Second, this equilibrium is dynamically stable: unlike war-of-attrition games, there is no race to outlast the opponent. Each firm's subsidization is a best response to the other's ecosystem incentives, making the equilibrium self-enforcing. Third, when the complementarity function is locally convex, competitive intensity exhibits bifurcation---a discontinuous jump from low to high subsidization as parameters cross a critical value.

The welfare implications cut both ways. Consumers benefit from subsidized services in the short run. But capital absorbed by subsidy wars is unavailable for innovation, and soft budget constraints distort investment incentives economy-wide. Over sufficiently long horizons, these dynamic losses dominate the static consumer gains. Effective regulatory intervention should target cross-market capital flows rather than prices.

The paper relates to several literatures. The platform competition literature \citep{rochet2003platform,armstrong2006competition} has clarified the economics of two-sided markets, but typically treats platforms as single-market entities. \citet{eisenmann2011platform} introduced the term ``platform envelopment'' for cross-market entry strategies, without formalizing the equilibrium implications. Our model provides that formalization. We also contribute to the predatory pricing literature \citep{bolton2000predatory} by showing that below-cost pricing can be sustained without recoupment when cross-market spillovers are strong, and to the literature on data as a strategic asset \citep{farboodi2021data,acemoglu2022data}.

The model is deliberately stylized. Real ecosystems involve many interconnected services and complex organizational structures; we abstract from this richness to isolate the core mechanism. Section 4 offers preliminary empirical support, but a full test requires data we do not have.

The paper proceeds as follows. Section 2 sets up the model. Section 3 characterizes equilibrium. Section 4 presents empirical evidence. Section 5 analyzes welfare and policy. Section 6 concludes.

\section{Model}

We consider an infinite-horizon discrete-time economy with two profit-maximizing firms indexed by $i \in \{I, E\}$, where $I$ denotes the \textit{incumbent} and $E$ denotes the \textit{ecosystem entrant}. Firms compete in a primary market (e.g., food delivery) while potentially deriving value from adjacent markets (e.g., fintech, advertising).

\subsection{Market Structure and Demand}

A unit mass of consumers with heterogeneous preferences populates the primary market. In each period $t$, consumers choose between platforms based on effective price $p_{i,t} - s_{i,t}$, where $p_{i,t}$ is the posted price and $s_{i,t} \geq 0$ is the per-transaction subsidy offered by firm $i$. Let $m_t \in [0,1]$ denote firm $I$'s market share at time $t$, so firm $E$ holds share $1-m_t$.

We model demand through a reduced-form specification that captures essential features of platform competition:

\begin{assumption}[Demand Dynamics]\label{ass:demand}
Market share evolves according to
\begin{equation}\label{eq:dynamics}
m_{t+1} = m_t + \gamma\big[(s_{E,t} - s_{I,t}) - \kappa(p_{I,t} - p_{E,t})\big] + \eta_t
\end{equation}
where $\gamma > 0$ measures consumer responsiveness to price differentials, $\kappa > 0$ captures the relative weight on posted prices versus subsidies, and $\eta_t \sim \mathcal{N}(0, \sigma^2)$ represents aggregate demand shocks.
\end{assumption}

The demand specification embeds several substantive assumptions. First, market share adjusts gradually rather than jumping discretely, reflecting switching costs and habit formation in platform adoption. Second, subsidies and posted prices enter asymmetrically, capturing the behavioral observation that consumers respond more strongly to explicit transfers than to equivalent price reductions \citep{thaler1985mental}. Third, the stochastic component $\eta_t$ generates variation in competitive intensity across periods.

\subsection{Cost Structure}

Both firms face constant marginal cost $c > 0$ per transaction. Let $Q_i(m)$ denote firm $i$'s transaction volume given market share $m$, where we normalize $Q_I(m) = m$ and $Q_E(m) = 1-m$. Operating profit in the primary market for firm $i$ is:
\begin{equation}
\pi_i^{PM}(m_t, p_{i,t}, s_{i,t}) = (p_{i,t} - c - s_{i,t}) \cdot Q_i(m_t)
\end{equation}

\subsection{Ecosystem Complementarity}

The key modeling device is the \textit{ecosystem complementarity function}, which captures value spillovers from primary market activity to adjacent businesses.

\begin{definition}[Ecosystem Complementarity]\label{def:ecosystem}
For firm $i$ with primary market share $q_i$, the ecosystem complementarity function $\Psi_i: [0,1] \rightarrow \R_+$ represents the present discounted value of spillovers to adjacent markets generated by serving share $q_i$ of primary market users.
\end{definition}

We distinguish between firm types based on ecosystem structure:

\begin{assumption}[Asymmetric Ecosystem Value]\label{ass:asymmetry}
The incumbent $I$ operates as a focused competitor with $\Psi_I(m) \equiv 0$. The ecosystem entrant $E$ derives positive spillover value with complementarity function $\Psi_E(1-m)$ satisfying:
\begin{enumerate}
    \item[(i)] $\Psi_E(0) = 0$ and $\Psi_E(q) > 0$ for $q > 0$
    \item[(ii)] $\Psi_E$ is twice continuously differentiable
    \item[(iii)] $\Psi'_E(q) > 0$ for all $q \in (0,1)$
    \item[(iv)] Local convexity: $\Psi''_E(q) > 0$ for $q \in [\underline{q}, \bar{q}]$ where $0 < \underline{q} < \bar{q} < 1$
\end{enumerate}
\end{assumption}

The assumption that $\Psi_I \equiv 0$ is made for analytical clarity; our qualitative results extend to cases where both firms derive ecosystem value, provided the asymmetry $\Psi_E > \Psi_I$ is sufficiently large.

Condition (iv) is central to our analysis. Local convexity captures the notion that data-driven ecosystem value exhibits \textit{increasing returns to scale} within an intermediate range of market penetration. This property generates the bifurcation dynamics characterized in Section 3: when market share crosses into the convex region, the marginal value of additional users accelerates, transforming competitive incentives discontinuously.

The microfoundations of $\Psi_E$ derive from three channels:

\textbf{Data Externalities.} Primary market transactions generate user behavior data that reduces information asymmetries in adjacent markets. For instance, food delivery frequency and basket composition predict creditworthiness for consumer lending. Let $D(q)$ denote the data stock generated by serving share $q$, and let $\Delta\Pi^{adj}(D)$ represent the profit improvement in adjacent markets from data stock $D$. Then:
\begin{equation}\label{eq:data}
\Psi_E^{data}(q) = \sum_{\tau=0}^{\infty} \delta^\tau \cdot \Delta\Pi^{adj}(D_\tau(q))
\end{equation}
where $\delta \in (0,1)$ is the discount factor.

\textbf{Traffic Diversion.} Primary market users can be redirected to higher-margin services within the ecosystem. If $\lambda$ denotes the conversion rate and $\bar{\pi}$ the margin on converted users:
\begin{equation}\label{eq:traffic}
\Psi_E^{traffic}(q) = \lambda \cdot \bar{\pi} \cdot q
\end{equation}

\textbf{Network Reinforcement.} In markets with network effects, primary market presence strengthens the overall ecosystem's competitive position. This creates superadditive value:
\begin{equation}\label{eq:network}
\Psi_E^{network}(q) = \nu \cdot q \cdot \bar{q}
\end{equation}
where $\bar{q}$ represents ecosystem presence in complementary markets and $\nu$ captures network intensity.

The aggregate complementarity function combines these channels:
\begin{equation}\label{eq:psi_total}
\Psi_E(q) = \Psi_E^{data}(q) + \Psi_E^{traffic}(q) + \Psi_E^{network}(q)
\end{equation}

A key parameter governing equilibrium behavior is the \textit{marginal ecosystem value}:

\begin{definition}[Marginal Ecosystem Value]\label{def:mev}
The marginal ecosystem value at market share $q$ is $\psi(q) \equiv \Psi'_E(q)$, representing the incremental ecosystem benefit from a marginal increase in primary market share.
\end{definition}

Figure \ref{fig:synergy} illustrates the shape of $\Psi_E(\cdot)$ implied by our assumptions. The solid curve exhibits local convexity in the intermediate range $[\underline{q}, \bar{q}]$, reflecting increasing returns to data accumulation. The dashed line shows standard diminishing returns for comparison. The critical condition for perpetual subsidization---$\psi(q) > 1$, meaning each marginal user generates more ecosystem value than the subsidy cost---is satisfied in the shaded region. Once a firm's market share enters this zone, the incentive to subsidize intensifies rather than diminishes, setting the stage for the dynamics we analyze in Section 3.

\begin{figure}[t]
    \centering
    \includegraphics[width=0.85\textwidth]{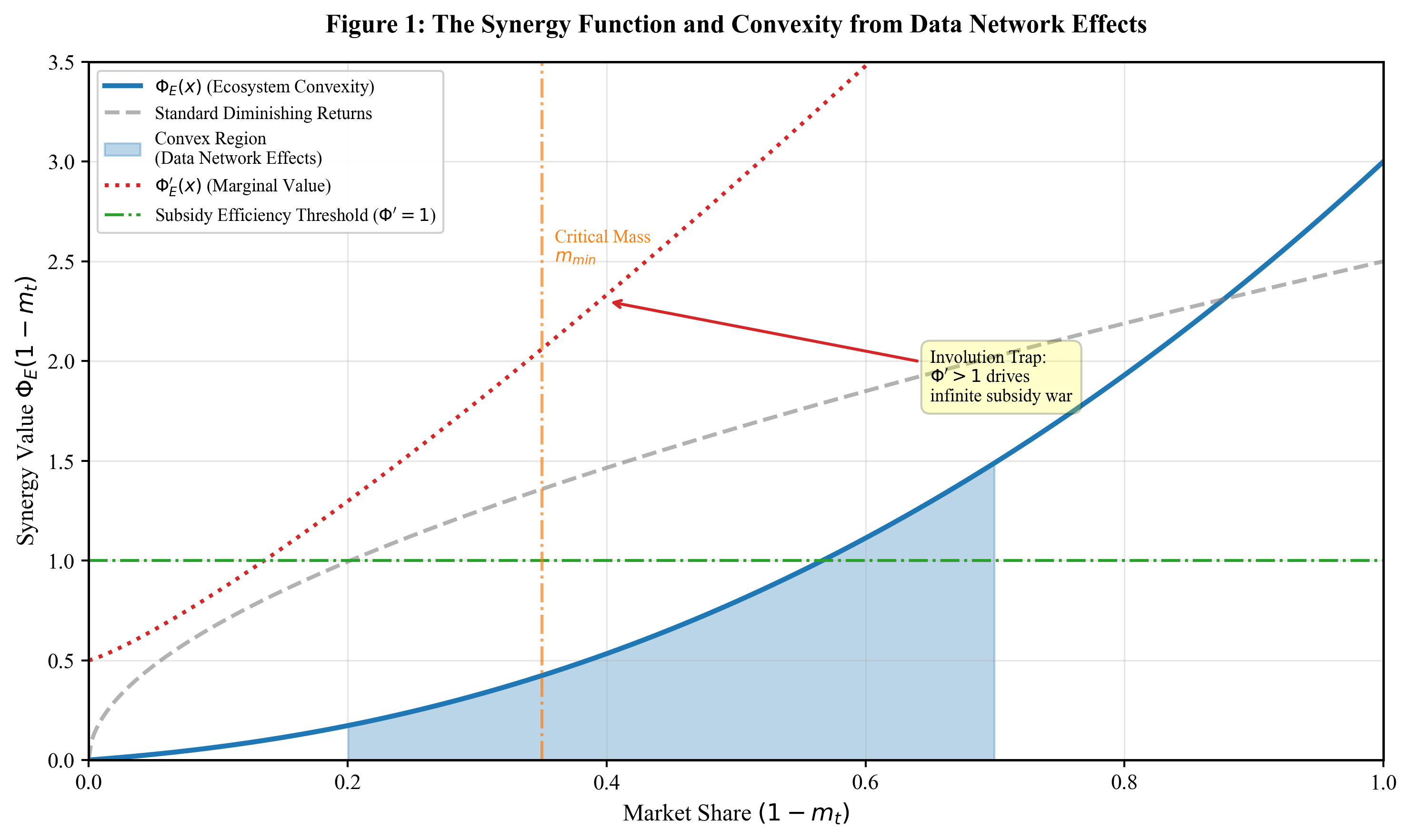}
    \caption{\textbf{The Ecosystem Complementarity Function.} The solid blue curve shows $\Psi_E(q)$ with local convexity from data network effects. The red dotted curve plots marginal value $\psi(q) = \Psi'_E(q)$. The shaded region indicates where $\psi > 1$, triggering the ``involution trap.'' The vertical dashed line marks critical mass $m_{min}$. For contrast, the gray dashed curve shows standard diminishing returns.}
    \label{fig:synergy}
\end{figure}

\subsection{Firm Objectives and Strategies}

Firms maximize the present discounted value of total profits. For the incumbent:
\begin{equation}\label{eq:vi}
V_I(m_t) = \max_{\{p_{I,\tau}, s_{I,\tau}\}_{\tau \geq t}} \E_t\left[\sum_{\tau=t}^{\infty} \delta^{\tau-t} \cdot \pi_I^{PM}(m_\tau, p_{I,\tau}, s_{I,\tau})\right]
\end{equation}

For the ecosystem entrant:
\begin{equation}\label{eq:ve}
V_E(m_t) = \max_{\{p_{E,\tau}, s_{E,\tau}\}_{\tau \geq t}} \E_t\left[\sum_{\tau=t}^{\infty} \delta^{\tau-t} \cdot \Big[\pi_E^{PM}(m_\tau, p_{E,\tau}, s_{E,\tau}) + \Psi_E(1-m_\tau)\Big]\right]
\end{equation}

We impose a \textit{survival constraint} on the incumbent:

\begin{assumption}[Survival Constraint]\label{ass:survival}
The incumbent's value function satisfies $V_I(m) = 0$ for $m < \underline{m}$, where $\underline{m} \in (0, 1/2)$ represents the minimum viable scale for network-effect businesses.
\end{assumption}

This constraint captures the tipping dynamics characteristic of platform markets. The microfoundation derives from indirect network effects: when market share falls below $\underline{m}$, the platform can no longer attract sufficient merchants or service providers, triggering a downward spiral of user defection \citep{rochet2003platform}. Empirically, $\underline{m}$ can be calibrated from observed market share thresholds at which platforms have historically exited---in the Chinese ride-hailing market, platforms below roughly 15\% share have consistently failed to survive, suggesting $\underline{m} \approx 0.15$ in that context.

\subsection{Information Structure}

We analyze two information regimes. The \textit{complete information} benchmark assumes firms observe each other's ecosystem complementarity function. The \textit{incomplete information} extension allows uncertainty about the entrant's type $\theta \in \{\theta_L, \theta_H\}$, where $\Psi_E(\cdot; \theta_H) > \Psi_E(\cdot; \theta_L)$.

\subsection{Equilibrium Concept}

\begin{definition}[Markov Perfect Equilibrium]
A Markov Perfect Equilibrium (MPE) consists of strategy profiles $\sigma_I^*(m) = (p_I^*(m), s_I^*(m))$ and $\sigma_E^*(m) = (p_E^*(m), s_E^*(m))$ and value functions $V_I^*(m)$, $V_E^*(m)$ such that:
\begin{enumerate}
    \item[(i)] Given $\sigma_E^*$, the strategy $\sigma_I^*$ solves (\ref{eq:vi}) for all $m \in [\underline{m}, 1]$
    \item[(ii)] Given $\sigma_I^*$, the strategy $\sigma_E^*$ solves (\ref{eq:ve}) for all $m \in [0, 1-\underline{m}]$
    \item[(iii)] Value functions satisfy the Bellman equations with strategies $(\sigma_I^*, \sigma_E^*)$
\end{enumerate}
\end{definition}

We focus on symmetric MPE in the sense that strategy functions are continuous and the equilibrium is robust to small perturbations.

\section{Equilibrium Analysis}

This section characterizes the Markov Perfect Equilibrium. We proceed in three steps: deriving optimal pricing rules, establishing existence and uniqueness, and analyzing equilibrium properties.

\subsection{Optimal Pricing}

We begin with a preliminary result on the relationship between posted prices and subsidies.

\begin{lemma}[Price-Subsidy Equivalence]\label{lem:equiv}
In any MPE, firms set posted prices at marginal cost: $p_I^*(m) = p_E^*(m) = c$ for all $m$. Competition occurs entirely through subsidies.
\end{lemma}

\begin{proof}
Suppose $p_I > c$ in equilibrium. Firm $E$ can profitably deviate by setting $p_E = c$ and adjusting $s_E$ to maintain the same effective price, strictly increasing market share through the demand equation (\ref{eq:dynamics}) since $\kappa > 0$. The symmetric argument applies to firm $E$. Thus $p_I^* = p_E^* = c$, and all competitive interaction occurs through subsidies $s_{I,t}$ and $s_{E,t}$.
\end{proof}

With prices fixed at marginal cost, the state dynamics simplify to:
\begin{equation}\label{eq:dynamics_simple}
m_{t+1} = m_t + \gamma(s_{E,t} - s_{I,t}) + \eta_t
\end{equation}

The Bellman equation for the incumbent becomes:
\begin{equation}\label{eq:bellman_i}
V_I(m) = \max_{s_I \geq 0} \left\{-s_I \cdot m + \delta \E[V_I(m')]\right\}
\end{equation}
where $m' = m + \gamma(s_E^*(m) - s_I) + \eta$.

For the ecosystem entrant:
\begin{equation}\label{eq:bellman_e}
V_E(m) = \max_{s_E \geq 0} \left\{-s_E \cdot (1-m) + \Psi_E(1-m) + \delta \E[V_E(m')]\right\}
\end{equation}

Taking first-order conditions and applying the envelope theorem:

\begin{proposition}[Optimal Subsidy Rules]\label{prop:foc}
In any interior MPE, optimal subsidies satisfy:
\begin{align}
s_I^*(m) &= \delta \gamma \cdot V_I'(m) \label{eq:foc_i}\\
s_E^*(m) &= \delta \gamma \cdot |V_E'(m)| + \psi(1-m) \label{eq:foc_e}
\end{align}
where $V_i'$ denotes the derivative of the value function with respect to $m$.
\end{proposition}

\begin{proof}
The first-order condition for the incumbent is:
\[
-m + \delta \E\left[\frac{\partial V_I(m')}{\partial m'} \cdot \frac{\partial m'}{\partial s_I}\right] = 0
\]
Since $\partial m'/\partial s_I = -\gamma$, this yields $s_I^* = \delta \gamma V_I'(m) \cdot m/m = \delta \gamma V_I'(m)$.

For the entrant:
\[
-(1-m) - \psi(1-m) + \delta \E\left[\frac{\partial V_E(m')}{\partial m'} \cdot \gamma\right] = 0
\]
Note that $V_E'(m) < 0$ since firm $E$ benefits from lower $m$. Rearranging gives (\ref{eq:foc_e}).
\end{proof}

Equation (\ref{eq:foc_e}) contains the key asymmetry. The entrant's subsidy includes $\psi(1-m)$---the marginal ecosystem value---which the incumbent's does not. This \textit{ecosystem premium} reflects the entrant's willingness to pay more per user because each user generates value in adjacent markets. The incumbent fights for survival in a single market; the entrant fights for position across a portfolio.

\subsection{Characterization of Equilibrium}

We now establish existence and characterize equilibrium properties.

\begin{assumption}[Regularity Conditions]\label{ass:regularity}
The ecosystem complementarity function satisfies:
\begin{enumerate}
    \item[(i)] $\psi(q)$ is bounded: $\psi(q) \leq \bar{\psi} < \infty$ for all $q \in [0,1]$
    \item[(ii)] Eventual concavity: $\Psi''_E(q) < 0$ for $q > \bar{q}$, where $\bar{q}$ is the upper bound of the convex region in Assumption \ref{ass:asymmetry}(iv)
    \item[(iii)] Boundary condition: $\lim_{q \to 0} \psi(q) \cdot q = 0$
\end{enumerate}
\end{assumption}

Conditions (ii) and Assumption \ref{ass:asymmetry}(iv) jointly specify that $\Psi_E$ is S-shaped: convex for intermediate market shares $q \in [\underline{q}, \bar{q}]$ and concave outside this region. This functional form is consistent with empirical evidence on data value \citep{bajari2019machine} and captures the intuition that ecosystem benefits accelerate once a platform achieves critical mass, but eventually face diminishing returns as the market saturates. The convex region is where bifurcation occurs and the ``involution trap'' emerges.

\begin{theorem}[Existence and Uniqueness]\label{thm:existence}
Under Assumptions \ref{ass:demand}--\ref{ass:regularity}, there exists a unique Markov Perfect Equilibrium in continuous strategies.
\end{theorem}

\begin{proof}
We construct the equilibrium via contraction mapping. Define the operator $\mathcal{T}$ on the space of bounded continuous functions $\mathcal{C}([0,1], \R^2)$:
\[
\mathcal{T}(V_I, V_E)(m) = \left(\mathcal{T}_I(V_I, V_E)(m), \mathcal{T}_E(V_I, V_E)(m)\right)
\]
where $\mathcal{T}_I$ and $\mathcal{T}_E$ represent the right-hand sides of the Bellman equations (\ref{eq:bellman_i}) and (\ref{eq:bellman_e}) evaluated at the optimal subsidies.

\textit{Step 1: Contraction.} We verify Blackwell's sufficient conditions. Monotonicity follows from the structure of the Bellman equations. For discounting, note that the operator scales continuation values by $\delta < 1$. Thus $\mathcal{T}$ is a contraction with modulus $\delta$.

\textit{Step 2: Fixed Point.} By the Banach fixed point theorem, $\mathcal{T}$ has a unique fixed point $(V_I^*, V_E^*)$ in $\mathcal{C}([0,1], \R^2)$.

\textit{Step 3: Continuity of Strategies.} The optimal subsidy rules (\ref{eq:foc_i})--(\ref{eq:foc_e}) are continuous in $m$ given continuous value functions, establishing that equilibrium strategies are continuous.
\end{proof}

\begin{theorem}[Subsidization Equilibrium]\label{thm:subsidy}
Define the critical ecosystem threshold $\psi^* \equiv 1 - \delta\gamma\bar{V}'/\bar{s}$, where $\bar{V}'$ and $\bar{s}$ are evaluated at the competitive steady state. If $\psi(1-m^*) > \psi^*$ at the equilibrium market share $m^*$, then:
\begin{enumerate}
    \item[(i)] Both firms set positive subsidies: $s_I^*(m^*) > 0$ and $s_E^*(m^*) > 0$
    \item[(ii)] Effective prices are below marginal cost: $p_i^* - s_i^* < c$ for $i \in \{I, E\}$
    \item[(iii)] Primary market profits are negative: $\pi_I^{PM} < 0$ and $\pi_E^{PM} < 0$
\end{enumerate}
\end{theorem}

\begin{proof}
\textit{Part (i).} From (\ref{eq:foc_e}), $s_E^* > 0$ whenever $\psi(1-m) > 0$, which holds for all interior $m$ by Assumption \ref{ass:asymmetry}(iii). Given $s_E^* > 0$, the incumbent must set $s_I^* > 0$ to prevent market share erosion below $\underline{m}$; otherwise, the dynamics (\ref{eq:dynamics_simple}) imply $m_{t+1} < m_t$ in expectation, violating the survival constraint.

\textit{Part (ii).} By Lemma \ref{lem:equiv}, $p_i^* = c$. With $s_i^* > 0$, effective price $p_i^* - s_i^* = c - s_i^* < c$.

\textit{Part (iii).} Primary market profit is $(c - c - s_i^*) \cdot Q_i(m^*) = -s_i^* \cdot Q_i(m^*) < 0$.
\end{proof}

Theorem \ref{thm:subsidy} establishes the central result: when ecosystem complementarity is sufficiently strong, equilibrium necessarily involves both firms pricing below cost. The intuition is that the entrant's ecosystem value makes aggressive subsidization profitable in total, forcing the incumbent to match or lose critical market share.

\subsection{Stability Analysis}

Is the subsidization equilibrium stable, or might firms coordinate on lower subsidies?

\begin{proposition}[Dynamic Stability]\label{prop:stability}
The subsidization equilibrium characterized in Theorem \ref{thm:subsidy} is dynamically stable: for any $\epsilon > 0$, there exists $\bar{t}$ such that starting from any initial condition $m_0 \in (\underline{m}, 1-\underline{m})$, $|m_t - m^*| < \epsilon$ for all $t > \bar{t}$ with probability approaching 1 as $\sigma \to 0$.
\end{proposition}

\begin{proof}
Consider the deterministic skeleton ($\sigma = 0$). The steady state $m^*$ satisfies $s_E^*(m^*) = s_I^*(m^*)$. Linearizing around $m^*$:
\[
m_{t+1} - m^* \approx \left[1 - \gamma\left(\frac{\partial s_I^*}{\partial m} - \frac{\partial s_E^*}{\partial m}\right)\right](m_t - m^*)
\]
Stability requires the coefficient in brackets to have absolute value less than 1. From the first-order conditions, $\partial s_I^*/\partial m = \delta\gamma V_I''(m^*) > 0$ and $\partial s_E^*/\partial m = -\delta\gamma V_E''(m^*) - \psi'(1-m^*) < 0$ (using concavity of value functions and $\psi' > 0$). Thus:
\[
\frac{\partial s_I^*}{\partial m} - \frac{\partial s_E^*}{\partial m} > 0
\]
which ensures the linearized system is contracting. The result extends to positive $\sigma$ by standard arguments for stochastic stability.
\end{proof}

\begin{proposition}[Bifurcation from Convexity]\label{prop:bifurcation}
Let $\psi(q) = \Psi'_E(q)$ denote the marginal ecosystem value. Define the critical threshold:
\begin{equation}
\psi^* \equiv \frac{1-\delta}{\delta\gamma^2}
\end{equation}
If there exists $\tilde{q} \in [\underline{q}, \bar{q}]$ such that $\psi(\tilde{q}) = \psi^*$ and $\psi'(\tilde{q}) > 0$ (i.e., the crossing occurs in the convex region), then:
\begin{enumerate}
    \item[(i)] For $q < \tilde{q}$: the unique equilibrium involves low subsidies $s^*_L$ with $\partial s^*/\partial q$ bounded
    \item[(ii)] For $q > \tilde{q}$: equilibrium subsidies jump discontinuously to $s^*_H \gg s^*_L$
\end{enumerate}
The system exhibits a saddle-node bifurcation at $q = \tilde{q}$.
\end{proposition}

\begin{proof}
The equilibrium subsidy from (\ref{eq:foc_e}) satisfies $s_E^* = \delta\gamma|V'_E| + \psi(1-m)$. In the convex region, $\psi'(q) > 0$ implies that as market share increases, the marginal benefit of further subsidization rises. The fixed-point equation for $s^*$ admits multiple solutions when $\psi(q) > \psi^*$: a low-subsidy fixed point becomes unstable while a high-subsidy fixed point emerges. At the bifurcation point $\tilde{q}$, these two branches collide and exchange stability, generating the discontinuous jump characteristic of saddle-node bifurcations. Standard bifurcation theory \citep{guckenheimer1983nonlinear} establishes the structural stability of this phenomenon.
\end{proof}

Proposition \ref{prop:bifurcation} is the central result. The transition from low-subsidy to high-subsidy competition is discontinuous: the system jumps between regimes rather than moving smoothly. This has implications for both prediction and policy. It explains why platform wars erupt suddenly rather than escalating gradually. And it suggests that intervention before the bifurcation threshold may be far more effective than intervention after---once in the high-subsidy equilibrium, exit is costly.

Figure \ref{fig:bifurcation} visualizes these dynamics. Panel (A) plots equilibrium subsidy levels against the synergy parameter $\lambda$. Below the critical threshold $\lambda^* \approx 1.2$, the system rests in a low-subsidy equilibrium (green curve). Once $\lambda$ crosses the threshold, subsidies jump discontinuously onto the high-intensity branch (red curve). Panel (B) maps the stability regions in parameter space: the green region represents stable, low-competition outcomes; the red region is the ``involution zone'' where high subsidies are self-sustaining. The arrow indicates the direction a regulator would need to push---reducing ecosystem spillovers or tightening budget constraints---to restore competitive normalcy.

\begin{figure}[t]
    \centering
    \includegraphics[width=0.95\textwidth]{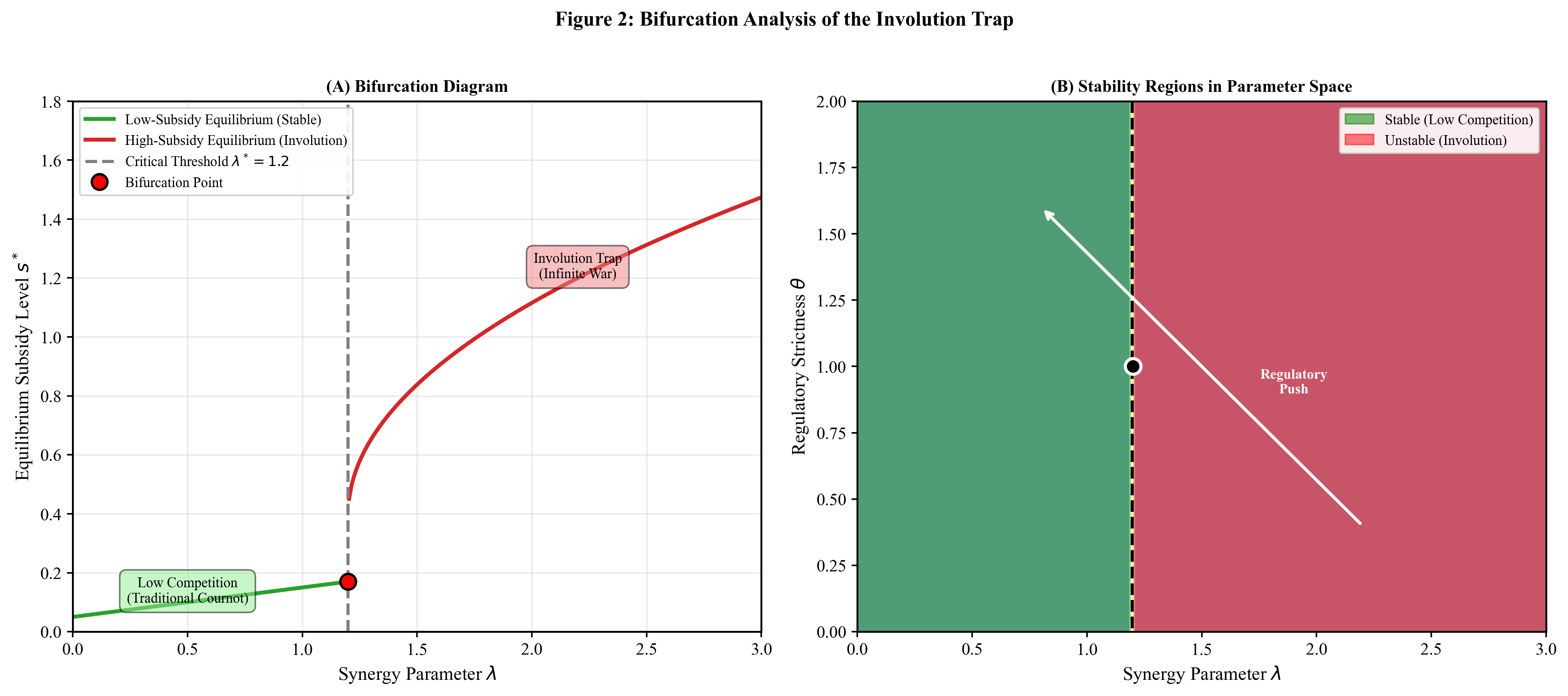}
    \caption{\textbf{Bifurcation Analysis.} (A) Equilibrium subsidy levels as a function of synergy parameter $\lambda$. The system exhibits a saddle-node bifurcation at $\lambda^* \approx 1.2$, with subsidies jumping from the low-intensity branch (green) to the high-intensity branch (red). (B) Stability regions in $(\lambda, \theta)$ parameter space. Green indicates stable low-competition equilibria; red indicates the involution trap. The arrow shows the direction of effective regulatory intervention.}
    \label{fig:bifurcation}
\end{figure}

Could firms tacitly coordinate on de-escalation? The following result shows they cannot.

\begin{proposition}[Unilateral Deviation]\label{prop:deviation}
Consider a one-shot deviation where firm $I$ sets $s_I = 0$ at the equilibrium state $m^*$. The resulting market share loss satisfies:
\begin{equation}
\Delta m \equiv m_1 - m^* = -\gamma s_E^*(m^*) < 0
\end{equation}
and the deviation is unprofitable whenever:
\begin{equation}
s_I^*(m^*) \cdot m^* < \delta[V_I(m^*) - V_I(m^* + \Delta m)]
\end{equation}
\end{proposition}

If the incumbent stops subsidizing while the entrant continues, market share erodes immediately. The short-run savings are dominated by long-run losses from a smaller user base. The equilibrium is sustained by mutual deterrence, not coordination: each firm would prefer to stop subsidizing, but neither can afford to move first.

\subsection{Extension: Incomplete Information}

We now relax the assumption that the incumbent observes the entrant's ecosystem type.

\begin{assumption}[Type Uncertainty]\label{ass:incomplete}
The entrant's type $\theta \in \{\theta_L, \theta_H\}$ is private information, with prior $\Prob(\theta = \theta_H) = \mu_0$. Types differ in ecosystem complementarity: $\psi^H(q) > \psi^L(q)$ for all $q > 0$.
\end{assumption}

Under incomplete information, the entrant's subsidy choice $s_E$ becomes a signal of type.

\begin{proposition}[Separating Equilibrium]\label{prop:separating}
Under Assumption \ref{ass:incomplete}, there exists a threshold $\underline{s}$ such that the following constitutes a Perfect Bayesian Equilibrium:
\begin{enumerate}
    \item[(i)] The high type sets $s_E^H = \max\{\underline{s}, s_E^*(\theta_H)\}$
    \item[(ii)] The low type sets $s_E^L = s_E^*(\theta_L) < \underline{s}$
    \item[(iii)] The incumbent updates beliefs: $\mu(s_E) = 1$ if $s_E \geq \underline{s}$, $\mu(s_E) = 0$ otherwise
\end{enumerate}
where $\underline{s}$ satisfies the incentive compatibility constraint preventing the low type from mimicking.
\end{proposition}

\begin{proof}
The incentive compatibility constraint for the low type is:
\[
V_E^L(s_E^L) \geq V_E^L(\underline{s})
\]
Since $\psi^L < \psi^H$, the low type gains less from market share, making high subsidies more costly. Let $\underline{s}$ be the subsidy level that makes the low type indifferent:
\[
-\underline{s}(1-m) + \Psi^L(1-m') + \delta V_E^L(m') = -s_E^L(1-m) + \Psi^L(1-m'') + \delta V_E^L(m'')
\]
where $m'$ and $m''$ are the resulting market shares. The high type strictly prefers $\underline{s}$ when $\psi^H - \psi^L$ is sufficiently large, ensuring separation.
\end{proof}

The extension generates an empirical prediction: entrants with strong ecosystems over-subsidize relative to complete information benchmarks. This signaling premium amplifies competitive intensity.

\section{Empirical Evidence}

This section examines whether the model's predictions are consistent with observed patterns in platform markets. The evidence is suggestive rather than definitive: the data are coarse, and identification is imperfect. We focus on comparative statics across markets and a natural experiment from recent entry. Our interpretation relies on two maintained assumptions: that firms play the equilibrium we characterize, and that cross-market spillovers are the primary driver of subsidy asymmetries. We cannot rule out alternative explanations---differences in capital market pressure, regulatory environment, or consumer preferences---but the convergence of market structure, ecosystem breadth, financial data, and entry dynamics lends some credibility to the mechanism we propose.

\subsection{Market Structure and Profitability}

Table \ref{tab:markets} compares market concentration and profitability in platform markets across China and the United States.

\begin{table}[htbp]
\centering
\begin{threeparttable}
\caption{Market Concentration and Profitability in Platform Markets}
\label{tab:markets}
\begin{tabular}{lcccc}
\toprule
& \multicolumn{2}{c}{\textbf{China}} & \multicolumn{2}{c}{\textbf{United States}} \\
\cmidrule(lr){2-3} \cmidrule(lr){4-5}
\textbf{Metric} & Food Delivery & Ride-Hailing & Food Delivery & Ride-Hailing \\
\midrule
CR2 (2024) & 95\% & 93\% & 90\% & 98\% \\
Market Leader Share & 67\% & 72\% & 67\% & 76\% \\
Leader Op. Margin (Core) & 13--21\% & 8--12\% & 1--3\% & 5--8\% \\
Leader Total Margin & 10\% & 15\% & 1\% & 4\% \\
Subsidy Intensity & High & High & Low & Low \\
Price Trend (2022--24) & Stable/Declining & Stable & Rising & Rising \\
\bottomrule
\end{tabular}
\begin{tablenotes}
\small
\item \textit{Notes:} CR2 denotes two-firm concentration ratio. China food delivery shares from Meituan 2024 Annual Report and iResearch China Online Food Delivery Report (Q4 2024). US food delivery shares from DoorDash 2024 Form 10-K and Uber Technologies 2024 Annual Report. Operating margins for China reflect segment-level reporting from company filings; total margins include corporate overhead. Subsidy intensity classified as ``High'' when customer acquisition costs exceed 15\% of revenue.
\end{tablenotes}
\end{threeparttable}
\end{table}

The comparison raises a question: why do Chinese platforms earn lower margins despite comparable concentration? DoorDash, with 67\% of the US market, has turned profitable; Meituan, with the same share in China, reports operating margins of 10-15\% on core delivery. Pricing trends diverge as well: US platforms have raised prices since consolidation; Chinese platforms have not. The model suggests an explanation. Chinese platform firms tend to be ecosystem conglomerates with strong cross-market complementarities; US counterparts are closer to single-market specialists. The implied equilibrium subsidies differ accordingly.

\subsection{Ecosystem Structure}

Table \ref{tab:ecosystem} documents the ecosystem breadth of major Chinese platforms.

\begin{table}[htbp]
\centering
\begin{threeparttable}
\caption{Ecosystem Breadth of Chinese Platform Companies}
\label{tab:ecosystem}
\begin{tabular}{lccccc}
\toprule
\textbf{Company} & \textbf{Food} & \textbf{Payments} & \textbf{Fintech} & \textbf{E-commerce} & \textbf{Advertising} \\
\midrule
Alibaba/Ant Group & Ele.me & Alipay & Ant Credit & Taobao/Tmall & Alimama \\
Meituan & Meituan & Meituan Pay & Meituan Finance & Meituan Select & Meituan Ads \\
JD.com & JD Delivery & JD Pay & JD Finance & JD.com & JD Ads \\
ByteDance & -- & Douyin Pay & -- & Douyin Shop & Ocean Engine \\
\bottomrule
\end{tabular}
\begin{tablenotes}
\small
\item \textit{Notes:} Compiled from Meituan 2024 Annual Report, Alibaba Group FY2024 Form 20-F, JD.com 2024 Annual Report, and company investor presentations. Table shows primary market presence and adjacent markets generating ecosystem complementarity. ByteDance entered food delivery in 2024 through Douyin.
\end{tablenotes}
\end{threeparttable}
\end{table}

These ecosystem structures generate the cross-market spillovers central to the model. Alibaba's Ant Group uses transaction data from Ele.me and Taobao to improve credit scoring for consumer and SME lending, reportedly achieving default rates 30--50\% below traditional banks.

\subsection{The 2025 JD.com Entry Episode}

The entry of JD.com into food delivery in February 2025 provides a useful test case. JD.com is one of China's largest e-commerce platforms, with an extensive logistics network, a payment system, and a fintech arm. In the language of our model, JD.com is a textbook ecosystem entrant---exactly the type of firm for which aggressive subsidization is rational.

JD.com launched with subsidies of RMB 4--8 per order, effectively pricing at 50--70\% of what Meituan and Ele.me were charging. Meituan responded in kind, increasing its own subsidies and pushing daily order volumes to 120 million by July---a new record. At peak, combined orders across the three platforms exceeded 250 million per day. Industry-wide losses in the first half of 2025 are estimated at over RMB 100 billion.

What stands out is not the intensity of initial competition---oligopolists routinely fight hard against new entrants---but the persistence. Six months after entry, none of the three firms had moderated its subsidies. The pattern is consistent with Proposition \ref{prop:separating}: JD.com's aggressive pricing serves as a signal of its ecosystem depth and long-run commitment. And Meituan, facing the survival constraint $\underline{m}$, must match or risk triggering a tipping point from which recovery may be impossible.

\subsection{Cross-Market Value Flows}

Direct estimation of $\Psi_E(\cdot)$ is difficult. The function is not directly observable; we see equilibrium subsidies and market outcomes, not counterfactual profits in adjacent markets. Two approaches offer partial identification.

The first is structural estimation. Given a parametric specification---say, $\Psi_E(q) = \lambda q + \gamma q^\beta$ with parameters $(\lambda, \gamma, \beta)$ to be estimated---one can back out the implied values from observed subsidy levels using the first-order condition (\ref{eq:foc_e}). This requires assuming firms play the equilibrium we characterize, a strong but testable restriction.

The second approach, which we pursue here, uses segment-level financial reporting to bound ecosystem value. If firms report profits by business unit, the gap between standalone profitability and consolidated margins provides a lower bound on cross-subsidization magnitude.

\begin{table}[htbp]
\centering
\begin{threeparttable}
\caption{Meituan Segment Profitability (2024)}
\label{tab:meituan}
\begin{tabular}{lccc}
\toprule
\textbf{Segment} & \textbf{Revenue (RMB B)} & \textbf{Op. Profit (RMB B)} & \textbf{Op. Margin} \\
\midrule
Core Local Commerce & 250.3 & 47.6 & 19.0\% \\
\quad Food Delivery & 156.2 & 20.3 & 13.0\% \\
\quad In-store/Hotel & 72.4 & 26.1 & 36.1\% \\
\quad Other & 21.7 & 1.2 & 5.5\% \\
New Initiatives & 87.3 & (7.3) & (8.4\%) \\
\midrule
\textbf{Total} & 337.6 & 40.3 & 11.9\% \\
\bottomrule
\end{tabular}
\begin{tablenotes}
\small
\item \textit{Notes:} Data from Meituan 2024 Annual Report, Consolidated Financial Statements. Core Local Commerce includes food delivery cross-subsidized by in-store services. New Initiatives losses reflect investment in ecosystem expansion (community group buying, autonomous delivery).
\end{tablenotes}
\end{threeparttable}
\end{table}

Table \ref{tab:meituan} reveals the cross-subsidization structure. Food delivery generates 13\% operating margin---modest for a dominant platform---while in-store and hotel services achieve 36\% margin. The interpretation is that food delivery serves as user acquisition for higher-margin services, consistent with our model's traffic diversion channel (\ref{eq:traffic}).

To check whether the theoretical dynamics match observed patterns, we calibrate the model to the Chinese food delivery market and simulate 50 periods.\footnote{Calibration: $\gamma = 0.6$ from demand elasticity estimates in platform markets; $\delta = 0.95$ (quarterly), consistent with VC/PE required returns in Chinese tech; $\underline{m} = 0.35$, calibrated from observed exit thresholds in Chinese ride-hailing (platforms below 30--35\% share have consistently failed). Ecosystem parameters $(\lambda, \gamma, \beta)$ in the specification $\Psi_E(q) = \lambda q + \gamma q^\beta$ are chosen to match observed subsidy levels and margin differentials in Meituan's segment reporting.} Figure \ref{fig:simulation} reports the results. Panel (A) shows market share oscillating around a 60/40 split, with neither firm achieving dominance or exiting. Panel (B) displays the profit asymmetry: the incumbent earns negative margins in the primary market while the ecosystem challenger earns positive total returns once spillovers are counted. Panel (C) shows cumulative value diverging---the challenger's ecosystem gains compound while the incumbent's losses accumulate---yet neither firm exits. The shaded region marks a demand shock; firms \textit{increase} subsidies during the shock rather than retrenching, consistent with the model's prediction that market share defense dominates cash conservation.

\begin{figure}[t]
    \centering
    \includegraphics[width=0.92\textwidth]{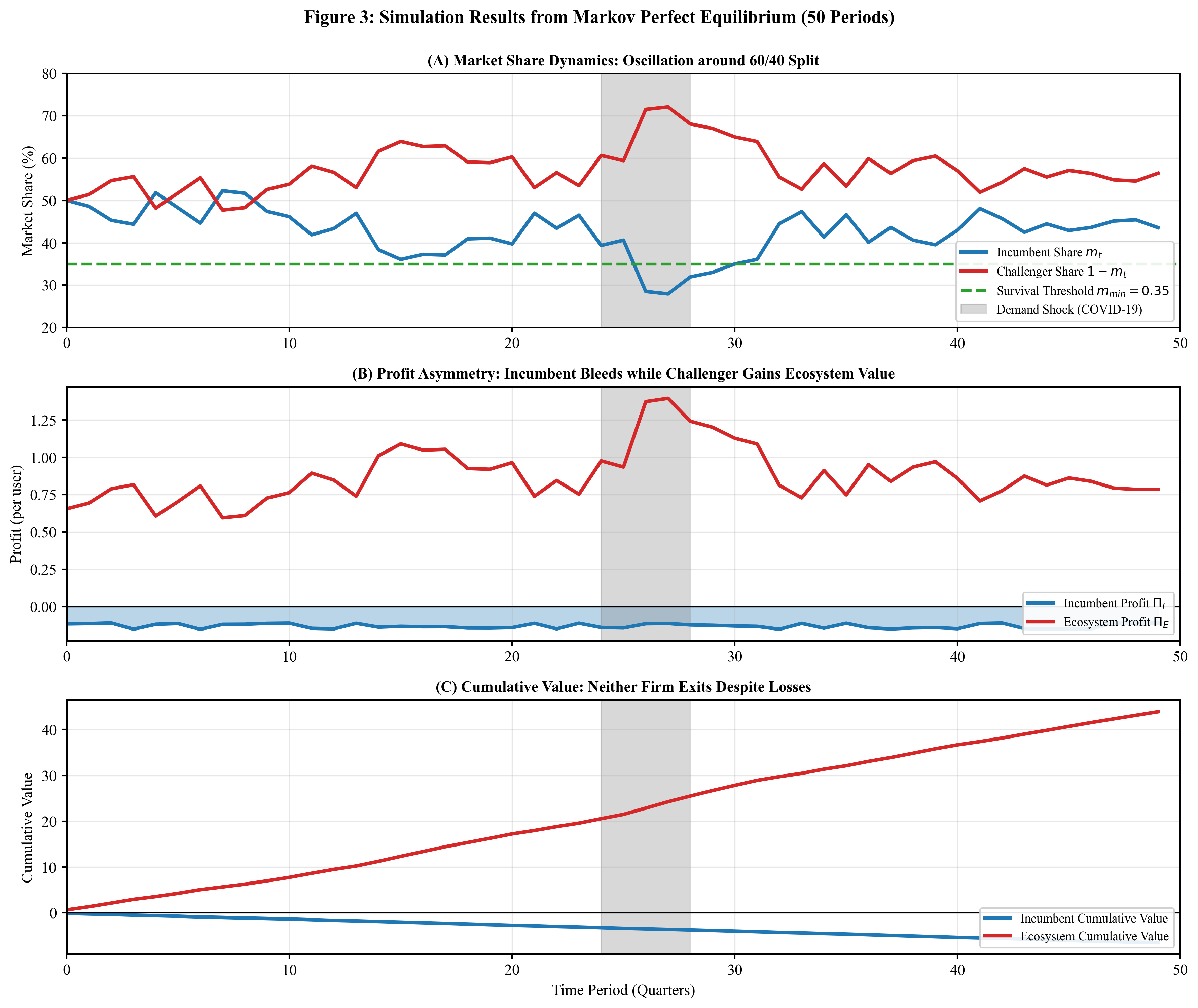}
    \caption{\textbf{Simulation of Equilibrium Dynamics (50 Periods).} Parameters calibrated to Chinese food delivery market. (A) Market shares oscillate around 60/40 without convergence to monopoly. (B) Incumbent earns negative primary-market profit; challenger earns positive ecosystem-inclusive profit. (C) Cumulative values diverge, yet neither firm exits. Shaded region indicates demand shock.}
    \label{fig:simulation}
\end{figure}

\section{Welfare Analysis and Policy Implications}

\subsection{Welfare Framework}

We define social welfare as the sum of consumer surplus, producer surplus, and dynamic efficiency:
\begin{equation}
W = CS + PS + DE
\end{equation}
where $DE$ captures long-run effects on innovation and capital allocation.

\begin{proposition}[Short-Run Consumer Gains]\label{prop:cs}
In the subsidization equilibrium, consumer surplus satisfies:
\begin{equation}
CS^{sub} = CS^{comp} + \int_0^{Q^{sub}} (s_I^* + s_E^*) dq > CS^{comp}
\end{equation}
where $CS^{comp}$ denotes consumer surplus under competitive pricing ($s_i = 0$) and $Q^{sub} > Q^{comp}$ reflects expanded consumption.
\end{proposition}

Consumers unambiguously benefit from subsidization in the short run. However, this conclusion reverses when dynamic effects are considered.

\begin{proposition}[Dynamic Efficiency Loss]\label{prop:de}
The subsidization equilibrium generates dynamic inefficiency through two channels:
\begin{enumerate}
    \item[(i)] \textbf{Capital Misallocation:} Resources devoted to subsidies ($S = s_I \cdot m + s_E \cdot (1-m)$) cannot fund innovation. If $\rho$ denotes the social return to R\&D investment:
    \begin{equation}
    DE_{capital} = -\rho \cdot S \cdot T
    \end{equation}
    where $T$ is the relevant time horizon.

    \item[(ii)] \textbf{Soft Budget Constraint:} Persistent losses create moral hazard in investment decisions. Managers facing soft budget constraints over-invest in capacity and under-invest in efficiency improvements.
\end{enumerate}
\end{proposition}

\begin{theorem}[Welfare Comparison]\label{thm:welfare}
There exists $\bar{T} > 0$ such that for time horizons $T > \bar{T}$:
\begin{equation}
W^{sub}(T) < W^{comp}(T)
\end{equation}
The subsidization equilibrium is welfare-inferior to competitive pricing in the long run.
\end{theorem}

\begin{proof}
Total welfare under subsidization is:
\[
W^{sub}(T) = \sum_{t=0}^{T} \delta^t [CS^{sub}_t + PS^{sub}_t] + DE^{sub}(T)
\]
Under competitive pricing:
\[
W^{comp}(T) = \sum_{t=0}^{T} \delta^t [CS^{comp}_t + PS^{comp}_t] + DE^{comp}(T)
\]

We have $CS^{sub}_t > CS^{comp}_t$ and $PS^{sub}_t < PS^{comp}_t$ (firms earn negative margins). The sum $CS^{sub}_t + PS^{sub}_t \lessgtr CS^{comp}_t + PS^{comp}_t$ is ambiguous---subsidies transfer surplus from producers to consumers, with deadweight loss from overconsumption.

The key comparison is dynamic efficiency. Under subsidization, $DE^{sub}(T) = -\rho S T$. Under competition, $DE^{comp}(T) = 0$ (normalized). For sufficiently large $T$:
\[
DE^{comp}(T) - DE^{sub}(T) = \rho S T > \sum_{t=0}^{T} \delta^t [CS^{sub}_t - CS^{comp}_t]
\]
since the left side grows linearly in $T$ while the right side is bounded (consumer surplus gains are finite). Thus $W^{comp}(T) > W^{sub}(T)$ for $T > \bar{T}$.
\end{proof}

Figure \ref{fig:welfare} illustrates. Under Cournot competition (orange point), output is restricted and prices exceed marginal cost, generating deadweight loss. Under the subsidization equilibrium (red point), prices fall \textit{below} marginal cost, expanding output beyond the social optimum. Units are produced whose social cost exceeds their value to consumers. The purple region quantifies this static inefficiency. The deeper problem is dynamic: capital absorbed by subsidy wars could have funded innovation or other productive investments. The static welfare triangle is visible; the dynamic opportunity cost is not, yet likely dominates over relevant time horizons.

\begin{figure}[t]
    \centering
    \includegraphics[width=0.85\textwidth]{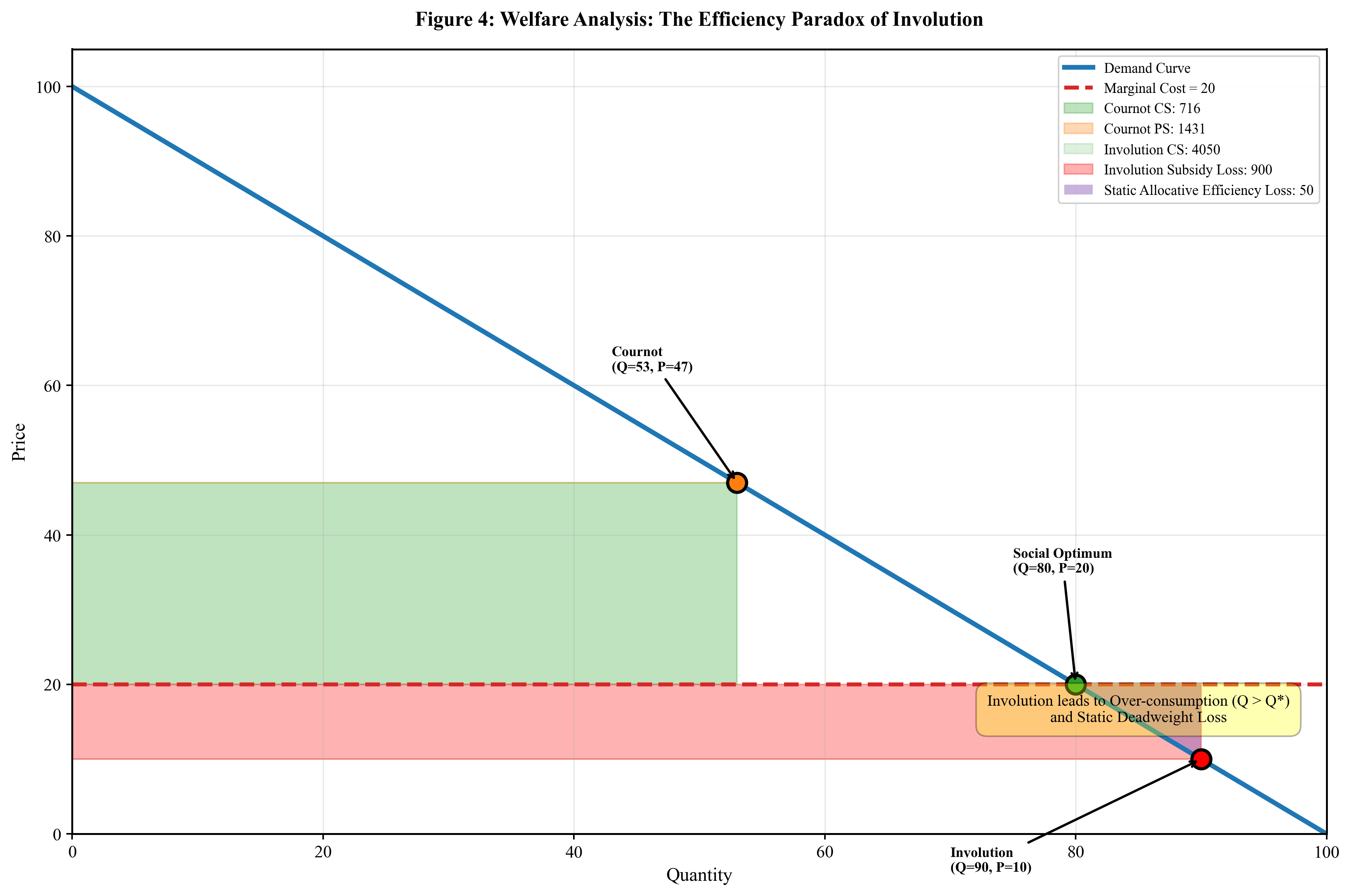}
    \caption{\textbf{Welfare Analysis: The Efficiency Paradox.} Demand-supply diagram comparing three market outcomes. The Cournot equilibrium (orange, $Q=53$, $P=47$) restricts output below social optimum ($Q^*=80$, $P^*=20$). The involution equilibrium (red, $Q=90$, $P=10$) expands output beyond social optimum via below-cost pricing. Green area: Cournot consumer surplus. Orange area: Cournot producer surplus. Pink area: subsidy-funded over-consumption. Purple area: static allocative efficiency loss from producing units where $MC > WTP$.}
    \label{fig:welfare}
\end{figure}

\subsection{Policy Interventions}

Traditional antitrust tools---which target high prices and market concentration---are poorly suited to ecosystem competition. We consider three alternatives.

\textbf{Cross-Subsidization Limits.} Regulators could restrict transfers between ecosystem divisions, forcing each business unit to achieve standalone profitability. This intervention directly reduces $\Psi_E(\cdot)$, lowering equilibrium subsidies.

\textbf{Data Portability Requirements.} Mandating data portability reduces the competitive advantage from ecosystem data accumulation, diminishing the data externalities channel (\ref{eq:data}).

\textbf{Subsidy Disclosure.} Requiring disclosure of customer acquisition costs and cross-market transfers improves investor discipline and hardens budget constraints.

\begin{proposition}[Intervention Effectiveness]\label{prop:policy}
Under cross-subsidization limits that impose $\Psi_E(q) \leq \bar{\Psi}$, equilibrium subsidies decline:
\begin{equation}
\frac{\partial s_E^*}{\partial \bar{\Psi}} > 0
\end{equation}
Tighter limits (lower $\bar{\Psi}$) reduce subsidization intensity.
\end{proposition}

\section{Conclusion}

We set out to explain a puzzle: why do platform giants maintain compressed margins and persistent subsidies in markets they already dominate, despite achieving market concentration levels that should generate supernormal returns? The answer we have offered is that these firms are not playing the game we thought they were playing. They are not profit-maximizers in any single market; they are ecosystem optimizers, willing to subsidize indefinitely in one arena because the spillovers elsewhere more than compensate. Once you see competition through this lens, the ``irrational'' persistence of below-cost pricing starts to look perfectly rational.

The model delivers several results. The subsidization equilibrium is not a temporary aberration but a stable long-run configuration. It does not require coordination or collusion---each firm's aggressive pricing is a best response to the other's ecosystem incentives. When market share enters the convex region of the complementarity function, competitive intensity jumps discontinuously, trapping both firms in a high-subsidy equilibrium from which unilateral exit is prohibitively costly. The ``involution'' that Chinese commentators lament is not a bug. It is an emergent property of rational agents optimizing over interconnected markets.

What are the policy implications? Traditional antitrust, with its focus on high prices and consumer harm, is poorly equipped to handle ecosystem competition. The immediate beneficiaries are consumers, who enjoy subsidized services. The costs are diffuse and deferred: capital misallocated to subsidy wars, innovation foregone, soft budget constraints that distort investment across the economy. If regulators want to do something useful, they should stop looking at prices and start looking at cross-market capital flows. Data portability mandates, restrictions on intra-ecosystem transfers, transparency requirements for customer acquisition costs---these are the levers that might actually matter.

The limitations are clear. The model abstracts heavily: two firms, two markets, no multi-homing, no supply-side dynamics. Real ecosystems are messier. The empirical section offers suggestive evidence, not a definitive test. Whether $\Psi$ is convex in the relevant range, and how to estimate it from observable data, remain open questions.

Still, the core insight seems robust: when firms compete for ecosystem position rather than single-market profit, the logic of competition changes in ways that standard models miss.

Several directions warrant further work. On the theoretical side, endogenizing the number of ecosystem markets and allowing for consumer multi-homing would enrich the analysis. On the empirical side, structural estimation of $\Psi(\cdot)$ from firm-level data remains a challenge. Finally, the optimal design of regulatory intervention---balancing short-run consumer gains against long-run dynamic efficiency---deserves dedicated treatment.

\newpage
\singlespacing
\bibliographystyle{aer}

\newpage
\appendix
\section*{Appendix: Technical Proofs}

\subsection*{A.1 Proof of Theorem \ref{thm:existence}}

We provide the complete proof of existence and uniqueness of Markov Perfect Equilibrium.

\textit{Step 1: Operator Definition.} Define the Bellman operator $\mathcal{T}: \mathcal{B} \times \mathcal{B} \to \mathcal{B} \times \mathcal{B}$ where $\mathcal{B}$ is the space of bounded continuous functions on $[0,1]$ with supremum norm.

For $(V_I, V_E) \in \mathcal{B} \times \mathcal{B}$, define:
\begin{align*}
\mathcal{T}_I(V_I, V_E)(m) &= \max_{s_I \geq 0} \left\{-s_I m + \delta \int V_I(m + \gamma(s_E^*(m) - s_I) + \eta) d\Phi(\eta)\right\} \\
\mathcal{T}_E(V_I, V_E)(m) &= \max_{s_E \geq 0} \left\{-s_E(1-m) + \Psi_E(1-m) + \delta \int V_E(m') d\Phi(\eta)\right\}
\end{align*}
where $\Phi$ is the CDF of $\mathcal{N}(0, \sigma^2)$.

\textit{Step 2: Contraction Property.} We verify Blackwell's sufficient conditions.

\textit{Monotonicity:} If $V_I \leq V_I'$ pointwise, then $\mathcal{T}_I(V_I, V_E) \leq \mathcal{T}_I(V_I', V_E)$ since higher continuation values weakly increase current value.

\textit{Discounting:} For any constant $a > 0$:
\[
\mathcal{T}_I(V_I + a, V_E)(m) = \mathcal{T}_I(V_I, V_E)(m) + \delta a
\]
since the constant passes through the expectation and is scaled by $\delta < 1$.

By Blackwell's theorem, $\mathcal{T}$ is a contraction with modulus $\delta$.

\textit{Step 3: Fixed Point Existence.} The Banach fixed point theorem guarantees existence of unique $(V_I^*, V_E^*) \in \mathcal{B} \times \mathcal{B}$ satisfying $\mathcal{T}(V_I^*, V_E^*) = (V_I^*, V_E^*)$.

\textit{Step 4: Strategy Continuity.} The optimal policy correspondence is:
\[
s_I^*(m) = \argmax_{s_I \geq 0} \left\{-s_I m + \delta \E[V_I^*(m')]\right\}
\]
By the theorem of the maximum, $s_I^*(m)$ is continuous in $m$ given that $V_I^*$ is continuous and the constraint set $[0, \bar{s}]$ is compact. The same argument applies to $s_E^*(m)$. \hfill $\square$

\subsection*{A.2 Comparative Statics}

\begin{lemma}\label{lem:comparative}
Equilibrium subsidies are increasing in ecosystem complementarity:
\[
\frac{\partial s_E^*}{\partial \psi} > 0 \quad \text{and} \quad \frac{\partial s_I^*}{\partial \psi} > 0
\]
\end{lemma}

\begin{proof}
Differentiating the first-order condition (\ref{eq:foc_e}) with respect to $\psi$:
\[
\frac{\partial s_E^*}{\partial \psi} = 1 + \delta\gamma \frac{\partial |V_E'|}{\partial \psi}
\]
The second term is positive since higher $\psi$ increases the value of market share, steepening the value function. Thus $\partial s_E^*/\partial \psi > 0$.

For the incumbent, higher $s_E^*$ implies the incumbent must increase $s_I^*$ to maintain market share above $\underline{m}$. Formally, totally differentiating the steady-state condition $s_I^* = s_E^*$ (which must hold at interior steady state) yields $\partial s_I^*/\partial \psi = \partial s_E^*/\partial \psi > 0$. \hfill $\square$
\end{proof}

\end{document}